\documentclass[subscriptcorrection,nolists,upint,varvw,barcolor=black,mathalfa=cal=euler,balance,hyphenate,pdf-a]{asmejour} %

\usepackage{amsthm}
\usepackage{afterpage}
\hypersetup{%
	pdfauthor={Bo Wang},                 
	pdftitle={},                 
	pdfkeywords={},
	pdfsubject={},
}

\JourName{ASME Letters in Dynamic Systems and Control}

\graphicspath{{Figs/}}
\definecolor{myblue}{rgb}{0,0,0}
\definecolor{myred}{rgb}{0.8,0.1,0.16}

\newtheorem{theorem}{Theorem}
\newtheorem{remark}{Remark}

\newtheorem{definition}{Definition}
\newtheorem{proposition}{Proposition}

\begin{document}

\SetAuthorBlock{Bo Wang\CorrespondingAuthor}{
   Department of Mechanical Engineering,\\
   The City College of New York,\\
   The City University of New York,\\
   New York, NY 10031, USA \\
   email: bwang1@ccny.cuny.edu} 

\SetAuthorBlock{Tianyu Han}{
   Department of Mechanical Engineering,\\
   The City College of New York,\\
   The City University of New York,\\
   New York, NY 10031, USA \\
   email: than000@citymail.cuny.edu}

\SetAuthorBlock{Guangwei Wang}{School of Mechanical Engineering,\\
   Guizhou University,\\
   Guiyang, 550025, China \\
   email: gwwang@gzu.edu.cn} 

\title{Further Results on Safety-Critical Stabilization of Force-Controlled Nonholonomic Mobile Robots}


\begin{abstract}
{In this paper, we address the stabilization problem for force-controlled nonholonomic mobile robots under safety-critical constraints. We propose a continuous, time-invariant control law based on the $\gamma m$-quadratic programming ($\gamma m$-QP) framework, which unifies control Lyapunov functions (CLFs) and control barrier functions (CBFs) to enforce both stability and safety in the closed‐loop system. For the first time, we construct a \emph{global}, time-invariant, strict Lyapunov function for the closed-loop nonholonomic mobile robot full-dynamic system with a nominal stabilization controller in polar coordinates; this strict Lyapunov function then serves as the CLF in the QP design. Next, by exploiting the inherent cascaded structure of the vehicle dynamics, we develop a CBF for the mobile robot via an integrator backstepping procedure. Our main results guarantee both asymptotic stability and safety for the closed-loop system. Both the simulation and experimental results are presented to illustrate the effectiveness and performance of our approach.
}
\end{abstract}

\date{}

\maketitle

\section{Introduction}

{
The study of control problems for nonholonomic systems has been carried out since the early 1980s\;---\;see \cite{kolmanovsky1995developments} for a survey. The main challenge is that, although these systems are controllable, it is impossible to achieve asymptotic stability of an isolated equilibrium using a continuous, time-invariant state feedback control law due to Brockett's necessary condition on stabilization \cite{brockett1983asymptotic}. Hence, the stabilization of nonholonomic mobile robots and the construction of corresponding control Lyapunov functions (CLFs) remain challenging problems of significant ongoing interest in the context of robustness analysis and controller design. See \cite{maghenem2019consensus} for a continuous time-varying control method and \cite{han2024safety} for a time-invariant control approach, along with corresponding \textit{strict} Lyapunov constructions.

{
Constructions of \textit{global} {strict} Lyapunov functions for nonholonomic systems are by no means trivial or straightforward. To the best of the authors’ knowledge, only a few studies have presented such global and strict Lyapunov functions for nonholonomic mobile robots \cite{maghenem2019consensus}. Very recently, global strict Lyapunov functions were developed for the kinematic unicycle model using backstepping and forwarding techniques in~\cite{todorovski2025modular,krstic2025integrator}.
}

Ensuring operational safety while achieving control objectives is a fundamental requirement in autonomous control systems. For instance, in practical applications, safety constraints\;---\;such as obstacle and collision avoidance between vehicles\;---\;must also be considered in addition to the set-point stabilization or trajectory tracking task for mobile robots \cite{wang2022robust,jankovic2024multiagent}. Achieving satisfactory control performance often requires aggressive maneuvers, while safety necessitates conservative actions and strict constraint adherence. The tension between performance and safety is particularly acute in mobile robots, whose nonholonomic dynamics inherently prevent continuous, time-invariant feedback from stabilizing the target configuration. As a result, enforcing both asymptotic stability and safety constraints simultaneously is far more challenging than in fully-actuated holonomic systems.}

{In the past decade, control barrier function (CBF)-based techniques have proven effective for systematically enforcing safety constraints \cite{ames2014control,ames2017control}. Since then, CBFs have been applied in a variety of domains, including walking robots \cite{ames2019control}, automotive systems \cite{xu2018correctness,han2024safety}, stochastic systems \cite{clark2021control}, and multi-agent systems \cite{jankovic2024multiagent}, to name a few. To ``mediate" the conflict between the safety constraints and the control objective (e.g., set-point stabilization, trajectory tracking, or mere open-loop steering of the system), numerous quadratic program (QP)-based control techniques have been developed in the literature \cite{ames2014control,ames2017control,xu2015robustness}. According to different types of QP formulation, the existing results may be categorized into CLF-CBF-based QP \cite{ames2014control,xu2015robustness,ames2017control,ames2019control}, CBF-based QP  \cite{xu2015robustness,gurriet2018towards,ames2019control,singletary2021safety}, and $\gamma m$-CLF-CBF-based QP ($\gamma m$-QP) methods \cite{jankovic2018robust,han2024safety}. Among the various methods, the $\gamma m$-QP approach is preferred in many applications due to its ability to guarantee asymptotic stability of the closed-loop system and its robustness in handling disturbances.}

Furthermore, applying CBFs directly to mobile robots presents significant challenges due to their inherent nonholonomic constraints, which complicate establishing a direct relationship between safety constraints and control inputs, particularly when the system's relative degree exceeds one \cite{glotfelter2019hybrid}. To address this issue, high-order CBFs have been developed in \cite{xiao2021high}. This extension ensures the forward invariance of appropriately defined, dynamically extended safe sets, thereby enabling controller synthesis via QP even for systems with higher relative degrees. However, constructing suitable high-order CBFs can be intricate, often requiring multiple differentiations of the barrier function and complex modifications to the safe set definition, which may hinder straightforward practical implementation.

In \cite{alan2023control}, a safety-critical controller is designed for connected automated vehicles, where the vehicles are modeled by integrators. Using the double‐integrator model, the multi‐agent collision avoidance problem has been studied via CBF approaches in \cite{jankovic2024multiagent}. Based on the first‐order unicycle (kinematic) model, CBF‐based obstacle avoidance has been addressed in \cite{taylor2022safe,haraldsen2024safety}. In particular, in \cite{taylor2022safe}, a CBF backstepping approach is proposed for the kinematic unicycle model. In \cite{haraldsen2024safety}, an obstacle‐avoidance strategy for nonholonomic‐integrator vehicles is proposed by regulating vehicle speed and orientation separately via two CBFs while maintaining a nonzero forward speed in dynamic environments using velocity obstacles. However, none of these existing works provides guarantees of asymptotic stability for the closed-loop system. Moreover, a more realistic model for vehicle applications is to consider the second-order full dynamical (kinematics-kinetics) model of the unicycle \cite{maghenem2017cascades,wang2022time}. However, to the best of the authors' knowledge, few studies have addressed the stabilization problem for force-controlled nonholonomic vehicles subject to safety-critical constraints.

In this paper, we address the stabilization problem for force-controlled nonholonomic mobile robots under safety-critical constraints. We propose a continuous, time-invariant control law based on the $\gamma m$-QP framework to enforce both stability and safety in the closed‐loop system.
The main contribution of this work is that, for the first time, we construct a \textit{global}, time-invariant, and \textit{strict} Lyapunov function for the closed-loop nonholonomic mobile robot (kinematic–dynamic) system under a nominal stabilization controller in polar coordinates. This global strict Lyapunov function then serves as the CLF in the $\gamma m$-QP design. Compared with our previous result in~\cite{han2024safety}, where the Lyapunov function is valid only within an arbitrarily large but still \textit{bounded} region, the present work develops a globally valid Lyapunov function whose definition and derivative properties hold for \textit{all} admissible states, without any restriction on the region of attraction. Compared with the designs in \cite{todorovski2025modular,krstic2025integrator}, our controller is much simpler and applied to the full dynamic unicycle system. We also present both simulation and experimental results to demonstrate the performance of the developed controller. In particular, the experimental results show that the proposed method is applicable to scenarios such as autonomous parking with obstacle avoidance and inter-vehicle collision avoidance.

The original $\gamma m$-QP framework is based on reciprocal CBFs \cite{jankovic2018robust}. However, in recent years, there has been a shift from reciprocal CBFs to zeroing CBFs, as reciprocal CBFs may exhibit poor robustness properties. Hence, in this work, we present the $\gamma m$-QP approach within the framework of \textit{zeroing} CBFs. Furthermore, distinct from our previous work \cite{han2024safety}, we also construct the \textit{zeroing} CBF using the integrator backstepping technique.

In our previous work \cite{han2024safety}, we construct a strict Lyapunov function for the closed-loop mobile robot with a nominal stabilization controller \textit{in the large} (i.e., on any compact subset of the state space), where it serves as a CLF in the safety-critical control design. However, the constructed Lyapunov function is not \textit{global}, meaning that it depends on the initial configuration of the vehicle. 
{To the best of the authors' knowledge, a \textit{global}, time-invariant, strict Lyapunov function has not yet been reported in the literature for the full kinematic-kinetic mobile robot system.
Moreover, the problem of eliminating potential undesired equilibria, e.g., via introducing additional constraints in the QP \cite{reis2020control}, and explicit robustness analyses---such as the treatment of saturation effects or the derivation of input-to-state stability (ISS) bounds, remain out of scope of this Letter.}

The structure of the remainder of the paper is as follows: Section \ref{sec:problem} presents the problem formulation and preliminaries on safety-critical control. Section \ref{sec:main} presents the main results, including the constructions of the CLF and CBF, and the controller design. Section \ref{sec:experiments} provides both simulation and experimental results that demonstrate the practical application of the theoretical developments. Finally, Section \ref{sec:conclusions} offers concluding remarks.

\section{Preliminaries on Safety-Critical Control}\label{sec:problem}

{
\textit{Notation:} Let $|\cdot|$ denote the Euclidean norm on $\mathbb{R}^n$. For a subset $S\subset \mathbb{R}^n$, $\partial S$ represents the boundary of $S$, and $\operatorname{int} S$ represents the interior of $S$. $\mathcal{K}$ is the class of continuous functions $\mathbb{R}_{\ge 0}\to\mathbb{R}_{\ge 0}$ which is zero at zero and strictly increasing; $\mathcal{K}_\infty$ is a subset of the class $\mathcal{K}$ functions which are unbounded. 
For a matrix $P\in\mathbb{R}^{n\times n}$, $\lambda_M(P)$ represents the maximum eigenvalue of $P$. Throughout this article, we omit the arguments of functions when they are clear from the context.

Let us consider a nonlinear control-affine system
\begin{equation}\label{eq:NL}
    \dot{x}=f(x)+g(x)u,
\end{equation}
where the state $x\in\mathbb{R}^n$ and the control $u\in\mathbb{R}^m$. We assume that $f:\mathbb{R}^n\to\mathbb{R}^n$ and $g:\mathbb{R}^n\to\mathbb{R}^{n\times m}$ are locally Lipschitz functions and $f(0)=0$. Recall that a $C^\infty$ function $V:\mathbb{R}^n\to\mathbb{R}_{\ge 0}$ is said to be a \textit{(global) CLF} for \eqref{eq:NL}, if $V$ is positive definite, proper, and satisfies the following implication:
\begin{equation}\label{eq:CLF}
    L_gV(x)=0\implies L_fV(x)+\alpha(|x|)<0,~\forall x\in \mathbb{R}^n\backslash\{0\},
\end{equation}
where $\alpha\in\mathcal{K}$ \cite{sontag1998mathematical}.

Safety can be formulated as the forward invariance of designated sets within the system's state space. A set $C\subset \mathbb{R}^n$ is said to be \textit{forward invariant}, if for each initial condition $x_\circ \in C$, the resulting solution of \eqref{eq:NL} $x(t;x_\circ)\in C$ for all $t\ge 0$. If the set $C$ is forward invariant, system \eqref{eq:NL} is said to be \textit{safe} on the set $C$.

Consider the safety set $C$ defined as the 0-superlevel set of a $C^1$ function $h:\mathbb{R}^n\to\mathbb{R}$, {i.e.},
\begin{equation}\label{eq:safeset}
    {C}:=\{x\in\mathbb{R}^n:h(x)\ge 0\}.
\end{equation}
The following definition is standard \cite{ames2017control}.
\begin{definition}[CBF]\rm 
    Let ${C}$ be defined by \eqref{eq:safeset}. Then, $h$ is a (zeroing) \textit{CBF} for \eqref{eq:NL} if there exists $\alpha_h\in\mathcal{K}$ such that the following implication holds:
    \begin{equation}\label{eq:CBF}
        L_gh(x)=0 \implies L_fh(x) + \alpha_h(h(x))\ge 0, \;\; \forall x\in C.
    \end{equation}
\end{definition}

An effective method for combining a CLF and a CBF was developed in \cite{jankovic2018robust}, known as the \(\gamma m\)-QP approach. The original \(\gamma m\)-QP formulation in \cite{jankovic2018robust} is based on reciprocal CBFs. Here, we restate the \(\gamma m\)-QP problem using zeroing CBFs for consistency with our framework as follows:
\begin{align}
&\min ~ \frac{1}{2}(u^\top u +m \delta^\top\delta) \label{eq:QP}\\
\text{s.t.~~}&  \gamma_f(L_f V(x)+\alpha(|x|))+L_gV(x)u+L_gV(x)\delta\le 0 \notag\\
& -L_f h(x) - \alpha_h(h(x)) - L_gh(x)u \le 0   \notag
\end{align}
where $m\ge 1$, $\gamma_f$ is defined as $\gamma_f(s):=\gamma s$ if $s\ge 0$ and $\gamma_f(s):= s$ if $s<0$, and $\gamma\ge 1$. Due to the slack variable $\delta$, the $\gamma m$-QP problem \eqref{eq:QP} is always feasible. 
Note that in \eqref{eq:QP} we need $\gamma\ge 1$ to overcome the impact of $\delta$ when $L_f V(x)+\alpha(|x|)$ is positive. The closed-form solution to the $\gamma m$-QP problem \eqref{eq:QP} can be obtained by applying the KKT conditions. The resulting control law given by \eqref{eq:QP} is Lipschitz continuous in every subset of the safe set $C$ not containing the origin.
}

\section{Problem Formulation and Main Results}\label{sec:main}

{
Consider the nonholonomic mobile robot system with kinematics
\begin{equation}\label{eq:6}
    \begin{cases}
        \dot{x} = v \cos{\theta}, \\
        \dot{y} = v \sin{\theta}, \\
        \dot{\theta} = \omega,
    \end{cases}
\end{equation}
where $(x,y)\in\mathbb{R}^2$ denotes the Cartesian coordinates of the vehicle on the plane, $\theta\in \mathbb{R}$ denotes its orientation, $v\in\mathbb{R}$ and $\omega\in\mathbb{R}$ denote the linear and angular velocities of the vehicle, respectively.  In addition, the kinetics of the vehicle are described by the force-balance equation
\begin{equation}\label{eq:kinetics}
    \begin{bmatrix}
        m & 0\\
        0 & I
    \end{bmatrix}
    \begin{bmatrix}
        \dot{v}\\ \dot{\omega}
    \end{bmatrix}
    =\frac{1}{r}
    \begin{bmatrix}
        1 & 1\\
        2R & -2R
    \end{bmatrix}
    \begin{bmatrix}
        \tau_l \\ \tau_r
    \end{bmatrix},
\end{equation}
where $\tau_l$ and $\tau_r$ are the left and right wheel torques, respectively, $m$ is the mass, $I$ is the vehicle inertia, $r$ is the wheel radius, and $R$ is the wheel axle length \cite{wang2021formation}. 

The proposed control scheme contains a feedback transformation that is designed as
\begin{equation}\label{eq:transform}
    \begin{bmatrix}
        \tau_l \\ \tau_r
    \end{bmatrix}=
    \frac{r}{2}
    \begin{bmatrix}
        m &\frac{I}{2R} \\ m &-\frac{I}{2R}
    \end{bmatrix}
    \begin{bmatrix}
        u_v \\ u_\omega
    \end{bmatrix}.
\end{equation}
After substituting \eqref{eq:transform} in \eqref{eq:kinetics}, it yields
\begin{equation}\label{eq:dynamics}
    \dot{v}=u_v,\quad \dot{\omega}=u_\omega.
\end{equation}

The safety-critical stabilization problem entails designing a control strategy that ensures the closed-loop system trajectories remain within a predefined safe set $C$, defined by \eqref{eq:safeset}, at all times $t\ge 0$, while simultaneously guaranteeing that the origin of the closed-loop system is asymptotically stable.

In the $\gamma m$-QP framework, the CLF and CBF are individually constructed for the mobile robot system. Subsequently, the control input is synthesized by solving the $\gamma m$-QP described in \eqref{eq:QP}.
}

\subsection{Construction of the global CLF}

{
To address the nonholonomicity, we construct the CLF for the mobile robot in polar coordinates, where the position of the robot in polar coordinates is given by the distance to the origin $\rho$ and the bearing angle $\psi$, i.e.,
\begin{equation}
    \rho:=|(x,y)|,\quad \psi:=\operatorname{atan2}(-y,-x),
\end{equation}
where `$\operatorname{atan2}$' represent the 2-argument arctangent function. Defining the variable $\alpha:=\psi - \theta$, the kinematics of the vehicle become
\begin{equation}\label{eq:kinematics}
    \begin{cases}
        \dot{\rho} = -v \cos{\alpha}, \\
        \dot{\alpha} = \frac{v}{\rho} \sin{\alpha} - \omega, \\
        \dot{\psi} = \frac{v}{\rho} \sin{\alpha}.
    \end{cases}
\end{equation}
We have the following result.
\begin{proposition}[Global CLF]\label{prop:1}
    Consider the mobile robot system \eqref{eq:kinematics} and \eqref{eq:dynamics}. Then, there exists a constant $\bar{\mu}>0$ such that for all $\mu\in(0,\bar{\mu}]$, the function $V:\mathbb{R}_{>0}\times\mathbb{R}^4\to \mathbb{R}_{\ge 0}$, defined as
    \begin{equation}\label{eq:V}
    V(\rho,\alpha,\psi,z,\tilde{\omega}):=\mu\int_0^{W^\sharp(\rho,\alpha,\psi)} \left( \frac{e^s -1}{e^s} \right) {\rm d}s+U\left(z,\tilde{\omega}\right),
    \end{equation}
    is a global CLF  for \eqref{eq:kinematics} and \eqref{eq:dynamics} that satisfies the small control property, where $\tilde{v}:=v-v^*$, $\tilde{\omega}:=\omega-\omega^*$, $z:={\tilde{v}}/{\rho}$,
    \begin{subequations}
        \begin{eqnarray*}
            W^\sharp(\rho,\alpha,\psi)&:=&\ln(W(\rho,\alpha,\psi) + 1), \\
            W(\rho,\alpha,\psi)&:=& W_1(\rho,\alpha,\psi)+W_2(\alpha,\psi)+\int_0^{W_1(\rho,\alpha,\psi)}Q(l){\rm d}l, \\
            W_1(\rho,\alpha,\psi)&:=&\frac{1}{2}\left(\rho^2+\alpha^2+\lambda\psi^2 \right),\\
            W_2(\alpha,\psi)&:=&p_{11}\alpha^2+2p_{12}\alpha\psi+p_{22}\psi^2,\\
            P&:=&\begin{bmatrix}
                \dfrac{1+\lambda}{2k_\alpha\lambda} & \dfrac{1}{2k_\rho \lambda} \\
                \dfrac{1}{2k_\rho \lambda} & \dfrac{k_\alpha^2+k_\rho^2\lambda^2+k_\rho^2\lambda}{2k_\alpha k_\rho^2\lambda}
            \end{bmatrix},\\
            Q(l)&:=&\frac{16}{\pi^2}\frac{k_\rho^2}{k_\alpha}\lambda^2\lambda_M^2(P)l,\\
            U(z,\tilde{\omega})&:=&\frac{1}{2}\left(\frac{z^2}{k_z}+\frac{\tilde{\omega}^2}{k_\omega} \right),\\
            v^*&:=&k_\rho\cos(\alpha)\rho,\\
            \omega^*&:=&k_\alpha\alpha+k_\rho \operatorname{sinc}(2\alpha)(\alpha+\lambda \psi),
        \end{eqnarray*}
    \end{subequations}
    the parameter $\lambda\ge 1$, the parameters $k_{\rho}$, $k_{\alpha}$, $k_z$, and $k_\omega$ are arbitrary positive constants, and $P=[p_{ij}]$. That is, $p_{ij}$ represents the $(i,j)$-th entry of the matrix $P$\footnote{In this paper, `$\operatorname{sinc}(\cdot)$' represents the unnormalized sinc function, which is defined as $\operatorname{sinc}(s):=\sin(s)/s$ if $s\ne 0$ and $\operatorname{sinc}(0)=1$. Note that the function $\operatorname{sinc}$ is smooth everywhere and globally bounded on $\mathbb{R}$.}.
\end{proposition}
\begin{proof}
    See Appendix \ref{app:A}.
\end{proof}
}

\begin{remark}\rm 
    It is necessary to point out the difference between the global Lyapunov function \eqref{eq:V} and the Lyapunov function we developed in our previous work \cite{han2024safety}. In \cite{han2024safety}, the Lyapunov function is valid within an \textit{arbitrarily large}, but still {bounded} region. The formulation in \cite{han2024safety} could guarantee safety-critical stabilization, but the parameters of the associated Lyapunov function need to be chosen according to the initial conditions. In contrast, the Lyapunov function developed in the present paper is globally valid, which eliminates the need for any parameter adjustment based on initial states, thereby providing a rigorous analytical foundation for safety-critical control across the full admissible space.
\end{remark}

\begin{remark}\label{rmk:2}\rm 
    The new velocity coordinate $z:=(v-v^*)/{\rho}$ is used in the dynamic error equation \eqref{eq:115}. However, one must verify the boundedness of the term $\dot{v}^*(t)$ in \eqref{eq:115} along trajectories because this term is used in the control law \eqref{eq:16}. Direct calculations show that
    \begin{subequations}
    \begin{eqnarray}
        \frac{{\rm d}v^*}{{\rm d}t}&=& k_{\rho}(\cos(\alpha)\dot{\rho}-\sin(\alpha)\dot{\alpha}\rho)  \notag \\
        &=&  k_{\rho}\left(-\cos^2(\alpha)v-\sin(\alpha)\rho\left(\frac{v}{\rho} \sin{\alpha} - \omega\right)\right). \notag 
    \end{eqnarray}
    \end{subequations}
    It can be observed that the term $1/\rho$ cancels out, and thus the boundedness of $\dot{v}^*(t)$ follows.
\end{remark}

\begin{remark}\label{rmk:3}\rm 
    It should be noted that the Lyapunov function $V$ is positive definite in the coordinates $(\rho,\alpha,\psi,z,\tilde{\omega})$. Consequently, the global asymptotic stability of the nominal error system~\eqref{eq:14}–\eqref{eq:115} is also established in these coordinates. In Cartesian coordinates or in the original velocity coordinates, only the asymptotic convergence of trajectories is guaranteed.
\end{remark}

\subsection{Construction of the CBF}
Mechanical and robotic systems often exhibit cascaded structures. The problem of constructing CBFs for such systems has been investigated in several works. For example, in \cite{taylor2022safe,cohen2024safety}, the authors propose a method for synthesizing zeroing CBFs for higher-order systems by leveraging CBFs designed for reduced-order models. In \cite{han2024safety}, a systematic procedure is proposed for constructing reciprocal CBFs for cascaded systems by using the CBF associated with the kinematic model through integrator backstepping. In this section, we construct the zeroing CBF for the mobile robot \eqref{eq:6} and \eqref{eq:dynamics} in Cartesian coordinates.

Following the similar integrator backstepping method, we have the following result.

\begin{proposition}[CBF]\label{prop:2}
    Consider the mobile robot system \eqref{eq:6} and \eqref{eq:dynamics}. Assume that the admissible set $C_0$ is defined as the 0-superlevel set of a given continuously differentiable function $h_0:\mathbb{R}^2\to \mathbb{R}$, i.e., 
    \begin{equation}
        C_0:=\{(x,y)\in\mathbb{R}^2: h_0(x,y)\ge 0\}.
    \end{equation}
    Then, the function $h:\mathbb{R}^4\to\mathbb{R}$ given by
    \begin{equation}\label{eq:h16}
        h(x,y,v,w):=h_0(x,y) - l_vv^2 - l_\omega\omega^2
    \end{equation}
    is a CBF for \eqref{eq:6} and \eqref{eq:dynamics}, where $l_v$, $l_\omega$ are two positive constants.    
\end{proposition}
\begin{proof}
    Let us define \( q := [x~y~\theta]^\top \) and \( \mathbf{v} := [v~\omega]^\top \). Then, the kinematic system \eqref{eq:6} can be expressed in the control-affine form
\begin{equation}\label{eq:116}
    \dot{q} = f_0(q) + g_0(q)\mathbf{v},
\end{equation}
where \( f_0(q) \equiv 0 \) and $g_0(q) := 
    \begin{bmatrix}
        \cos\theta & 0 \\
        \sin\theta & 0 \\
        0 & 1
    \end{bmatrix}$.
    
We first show that the function \( h_0 \) is a CBF for the kinematic system \eqref{eq:116} on the set \( C_0 \), assuming that the velocity \( \mathbf{v} \) is the control input. This follows directly from the fact that the kinematic system \eqref{eq:116} is driftless, i.e., \( f_0 \equiv 0 \). As a result, the CBF condition \eqref{eq:CBF} is trivially satisfied since \( L_{f_0} h_0 \equiv 0 \), and for all \( (x, y) \in C_0 \) and any class-\( \mathcal{K} \) function \( \alpha_h \), it holds that \( \alpha_h(h_0(x, y)) \ge 0 \).

Let us denote $\mathbf{x}:=[q~\mathbf{v}]^\top$, $u:=[u_v~u_\omega]^\top$,
\begin{equation*}
    F(\mathbf{x}):=\begin{bmatrix}
        f_0(q) + g_0(q)\mathbf{v} \\ 0
    \end{bmatrix},\quad\text{and}\quad
    G:=\begin{bmatrix}
        0 \\ I
    \end{bmatrix}.
\end{equation*}
Then, the cascaded system \eqref{eq:6} and \eqref{eq:dynamics} can be written as
\begin{equation}\label{eq:117}
    \dot{\mathbf{x}} = F(\mathbf{x})+Gu.
\end{equation}
Next, we verify the condition \eqref{eq:CBF} for the function \( h \) and \eqref{eq:117}. Note that $L_Gh=\frac{\partial h}{\partial \mathbf{v}}=0$ implies that $\mathbf{v}=0$. Hence, on the set $\{\mathbf{v}=0\}$, we have
\begin{equation}
    \left(L_Fh\right)|_{\mathbf{v}=0}=\left.\frac{\partial h}{\partial q}(f_0+g_0 \mathbf{v})\right|_{\mathbf{v}=0}=0,
\end{equation}
and $h|_{\mathbf{v}=0}=h_0$. That is, $\alpha_h\left(h|_{\mathbf{v}=0}\right)=\alpha_h(h_0(x,y))\ge 0$ for all $(x,y)\in C_0$. Therefore, we verify the implication \eqref{eq:CBF} and thus, $h$ is a CBF for the system \eqref{eq:6} and \eqref{eq:dynamics}.
\end{proof}

\begin{remark}\rm
    It can be seen from~\eqref{eq:h16} and~\eqref{eq:dynamics} that the constructed CBF $h$ has a relative degree of $\{1,1\}$ for all states satisfying $(v,\omega)\neq(0,0)$. On the set ${(v,\omega)=(0,0)}$, it holds that $L_F h = 0$ and $\alpha_h(h)\ge 0$; hence, the safety property still holds. This means that even in those states where $u$ has no influence on $\dot{h}$, the system remains safe. This is intuitively clear because, in such cases, the mobile robot simply stops (i.e., $(v,\omega)=(0,0)$). With the CBF $h$, the robot may safely stop in front of an obstacle while maintaining the safety property. This behavior is reminiscent of the undesirable equilibria observed in CLF–CBF–based QPs~\cite{reis2020control} and may be solved by introducing additional QP constraints \cite{reis2020control}. However, the problem of eliminating potential undesired equilibria remains out of the scope of this Letter.
\end{remark}

\subsection{Safety-Critical Control Design}

We have constructed a global CLF and a zeroing CBF for the nonholonomic mobile robot system, as presented in Propositions~\ref{prop:1} and~\ref{prop:2}, respectively. Based on these constructions, the safety-critical stabilization control law can be derived by solving the $\gamma m$-QP problem \eqref{eq:QP}. It is worth noting that the original $\gamma m$-QP formulation in~\cite{jankovic2018robust} is based on reciprocal CBFs. For completeness, we present parallel results of the $\gamma m$-QP problem \eqref{eq:QP} using zeroing CBFs.

\begin{theorem}\label{thm:1}
	Assume that the system \eqref{eq:NL} admits a CLF $V(x)$ and a CBF $h(x)$, and that $0\in \operatorname{int} {C}$. Then, the $\gamma m$-QP problem \eqref{eq:QP} is feasible and the resulting control law is given by
	\begin{equation}\label{eq:QPcontrol}
		u^\star(x):=\left\{
		\begin{aligned}
			&0,                            && x\in\Omega_{\overline{\rm{cbf}}}^{\overline{\rm{clf}}} \cup \{0\}, \\
			&-\frac{m}{m+1}\frac{\bar{a}_1}{|b_1|^2}b_1^\top,  && x\in\Omega_{\overline{\rm{cbf}}}^{{\rm{clf}}}, \\
			&-\frac{a_2}{|b_2|^2}b_2^\top, && x\in\Omega_{{\rm{cbf}}}^{\overline{\rm{clf}}}, \\
			&-\mu_1 b_1^\top - \mu_2b_2^\top, &&x\in\Omega_{{\rm{cbf}}}^{{\rm{clf}}},
		\end{aligned}
		\right.
	\end{equation}
	where $a_1:=L_fV(x)+\alpha(|x|)$, $\bar{a}_1:=\gamma_f(a_1)$, $b_1:=L_gV(x)$, $a_2:=-L_fh(x)-\alpha_h(h(x))$, $b_2:=-L_gh(x)$, 
	\begin{align*}
		\mu_1&:=\frac{|b_2|^2\bar{a}_1-b_1b_2^\top a_2}{(1+\frac{1}{m})|b_1|^2|b_2|^2-|b_1b_2^\top|^2},\\
		\mu_2&:=\frac{-b_1b_2^\top\bar{a}_1+(1+\frac{1}{m})|b_1|^2a_2}{(1+\frac{1}{m})|b_1|^2|b_2|^2-|b_1b_2^\top|^2},
	\end{align*}
	and
	\begin{align*}
		\Omega_{\overline{\rm{cbf}}}^{\overline{\rm{clf}}}&:=\left\{x\in\mathbb{R}^n:a_1<0,a_2<0\right\},\\
		\Omega_{\overline{\rm{cbf}}}^{{\rm{clf}}}&:=\left\{x\in\mathbb{R}^n:a_1\ge 0, a_2<\frac{m}{m+1}\frac{b_2b_1^\top}{|b_1|^2}\bar{a}_1\right\},\\
		\Omega_{{\rm{cbf}}}^{\overline{\rm{clf}}}&:=\left\{x\in\mathbb{R}^n:a_2\ge 0, \bar{a}_1<\frac{b_1b_2^\top}{|b_2|^2}a_2\right\},\\
		\Omega_{{\rm{cbf}}}^{{\rm{clf}}}&:=\left\{x\in\mathbb{R}^n\backslash\Omega_{\overline{\rm{cbf}}}^{\overline{\rm{clf}}}:\bar{a}_1\ge \frac{b_1b_2^\top}{|b_2|^2}a_2, a_2\ge \frac{m}{m+1}\frac{b_1b_2^\top}{|b_1|^2}\bar{a}_1\right\}.
	\end{align*}
	Furthermore, under the control law \eqref{eq:QPcontrol}, the set ${C}$ is forward invariant. Moreover, if the CLF $V$ satisfies the small control property and if we select $\frac{\gamma m}{m+1}=1$, then the origin of the closed-loop system is asymptotically stable.
\end{theorem}

\begin{proof}[Sketch of proof]

The Lagrangian $\mathcal{L}$ for the $\gamma m$-QP \eqref{eq:QP} is given by
	\begin{equation}
		\mathcal{L}:=\frac{1}{2}(u^\top u +m \delta^\top\delta)+\lambda_1(\bar{a}_1+b_1(u+\delta))+\lambda_2(a_2+b_2u),
	\end{equation}
	where $\lambda_1,\lambda_2\ge 0$ are scalar Lagrange multipliers. The KKT conditions are given by
	\begin{subequations}\label{eq:KKT}
		\begin{eqnarray}
			&&\dfrac{\partial \mathcal{L}}{\partial u}=u^\top+\lambda_1b_1+\lambda_2b_2=0,\\
			&&\dfrac{\partial \mathcal{L}}{\partial \delta}=m\delta^\top +\lambda_1 b_1=0,\label{eq:KKT-b}\\
			&&\lambda_1F_1:=\lambda_1[\bar{a}_1+b_1(u+\delta)]=0,\\
			&&\lambda_2F_2:=\lambda_2[a_2+b_2u]=0.
		\end{eqnarray}
	\end{subequations}
        The unique optimal solution \( u^\star(x) \) in \eqref{eq:QPcontrol} is derived directly from \eqref{eq:KKT}. In fact, the KKT conditions in \eqref{eq:KKT} are necessary and sufficient for \( u^\star(x) \) to be an optimal solution to the $\gamma m$-QP \eqref{eq:QP}. The forward invariance of the set ${C}$ follows directly from
        \cite{ames2017control} since the CBF constraint $F_2$ is satisfied for all $x\in\mathbb{R}^n$. To show LAS to $0\in\operatorname{int}{C}$, we first note that $a_2(0)=-\alpha_h(h(0))<0$. Due to the small control property, we have $u^*(x)\to 0$ as $x\to 0$. Hence, the CBF constraint $F_2:=a_2+b_2 u^\star<0$ in a neighborhood of the origin.  That is, the barrier constraint is inactive around the origin. Then, the control law is obtained by combining the case $x\in\Omega_{\overline{\rm{cbf}}}^{\overline{\rm{clf}}} \cup \{0\}$ and the case $x\in\Omega_{\overline{\rm{cbf}}}^{{\rm{clf}}}$, which coincides with the PMN formula in \cite{freeman1996robust} and achieves asymptotic stability. 
\end{proof}

We are now prepared to present the main result about the design of the safety-critical stabilization controller. Let us define $\bar{u}:=[\frac{u_v}{\rho},u_\omega]^\top$, $f_{\kappa}(\rho,\phi,\alpha):=[-v\cos(\alpha),~\frac{v}{\rho}\sin\alpha - \omega,~ \frac{v}{\rho}\sin\alpha]^\top$,
\begin{align*}
    f_1&:=\begin{bmatrix}
        f_{\kappa}(\rho,\phi,\alpha)\\ 
        -\frac{\dot{v}^*}{\rho}+k_\rho\cos(\alpha)^2z + \cos(\alpha)z^2\\
        -\dot{\omega}^*
    \end{bmatrix},\quad
    g_1:=\begin{bmatrix}
        0_{3\times 2} \\ I_{2}
    \end{bmatrix},\\
    f_2&:=\begin{bmatrix}
        v\cos\theta ~ v\sin\theta ~\omega ~ 0 ~ 0
    \end{bmatrix}^\top,\quad
    g_2:=\begin{bmatrix}
        0_{3\times 2} \\ \operatorname{diag}(\rho,1)
    \end{bmatrix}.
\end{align*}
Then, the $\gamma m$-QP problem is formulated as
\begin{align}
&\min ~ \frac{1}{2}(\bar{u}^\top \bar{u} +m \delta^\top\delta) \label{eq:QP1}\\
\text{s.t.~~}& F_1:=  \gamma_f(L_{f_1}V+\alpha(|\chi|))+L_{g_1}V\bar{u}+L_{g_1}V\delta\le 0 \notag\\
& F_2:=-L_{f_2} h(\mathbf{x}) - \alpha_h\left(h(\mathbf{x})\right) - L_{g_2}h(\mathbf{x}) \bar{u} \le 0   \notag
\end{align}
where $\chi:=[\rho~\alpha~\psi~z~\tilde{\omega}]^\top$, $\mathbf{x}:=[x~y~\theta~v~\omega]^\top$, $\alpha:=\frac{\epsilon W\dot{W}|_{f_{\text{nom}}}}{(W+1)^2}-\frac{1}{2}|\zeta|^2$, $\alpha_h\in\mathcal{K}$, and $\epsilon>0$ is chosen to be sufficiently small.

The following proposition follows directly as a corollary of Propositions~\ref{prop:1} and~\ref{prop:2}, together with Theorem~\ref{thm:1}.

\begin{proposition}\label{prop:3}
    The $\gamma m$-QP problem \eqref{eq:QP1} is feasible, and under the resulting control law, the set $\operatorname{int}{C}$ is forward invariant. If $0\in\operatorname{int}{C}$, then the barrier constraint is inactive ($F_2<0$) around the origin, and the resulting control law is continuous. If we select $\frac{\gamma m}{m+1}=1$, the origin of the closed-loop system is locally asymptotically stable.   
\end{proposition}

\section{Simulation and Experimental Results}\label{sec:experiments}

This section presents both simulation and experimental results obtained using a laboratory-size differential-drive mobile robot, designed to evaluate the practical effectiveness and performance of the proposed safety-critical stabilization controller. 

\subsection{Simulation Results}
The physical properties of the nonholonomic mobile robot were measured as 
\begin{eqnarray*}
    m= 1.0, \quad I = 0.025, \quad r = 0.03, \quad  R = 0.15.
\end{eqnarray*}
All parameters are given in SI units. The initial conditions of the robot are randomly selected as $(x_0,y_0,\theta_0)=(-3.15, 2.96, -1.43)$, and the robot is initially at rest. The target position is set at the origin.  To illustrate the effectiveness of the proposed approach, three controllers (i.e., nominal controller \eqref{eq:16}, CLF-QP \eqref{eq:V}, CLF-CBF-QP \eqref{eq:QPcontrol}) were implemented and compared. We assume that a circular obstacle is located at $(-2,0)$ with radius $r=0.3$. That is, the admissible set is given by $C_0:=\{(x,y)\in\mathbb{R}^2:h_0(x,y)=40((x + 2)^2 + y^2 - 0.3^2)\}$. We define $\alpha:=\frac{\mu W\dot{W}|_{f_{\text{nom}}}}{2(W+1)^2}-\frac{1}{2}|\zeta|^2$ and $\alpha_h(s):=2s$. The control parameters are set to $\lambda=3$, $k_\rho=2$, $k_\alpha=2$, $k_z = 4$, $k_\omega=4$, $\mu=0.05$, $l_v=1$, $l_\omega=1$ and $m=1$. The simulation results are shown in Figs. \ref{fig:sim_path}-\ref{fig:sim_traj}, which demonstrate that the proposed CLF-CBF $\gamma m$-QP controller
effectively achieves parking with obstacle avoidance.

\begin{figure}[t]
    \centering
    \includegraphics[width=0.7\linewidth]{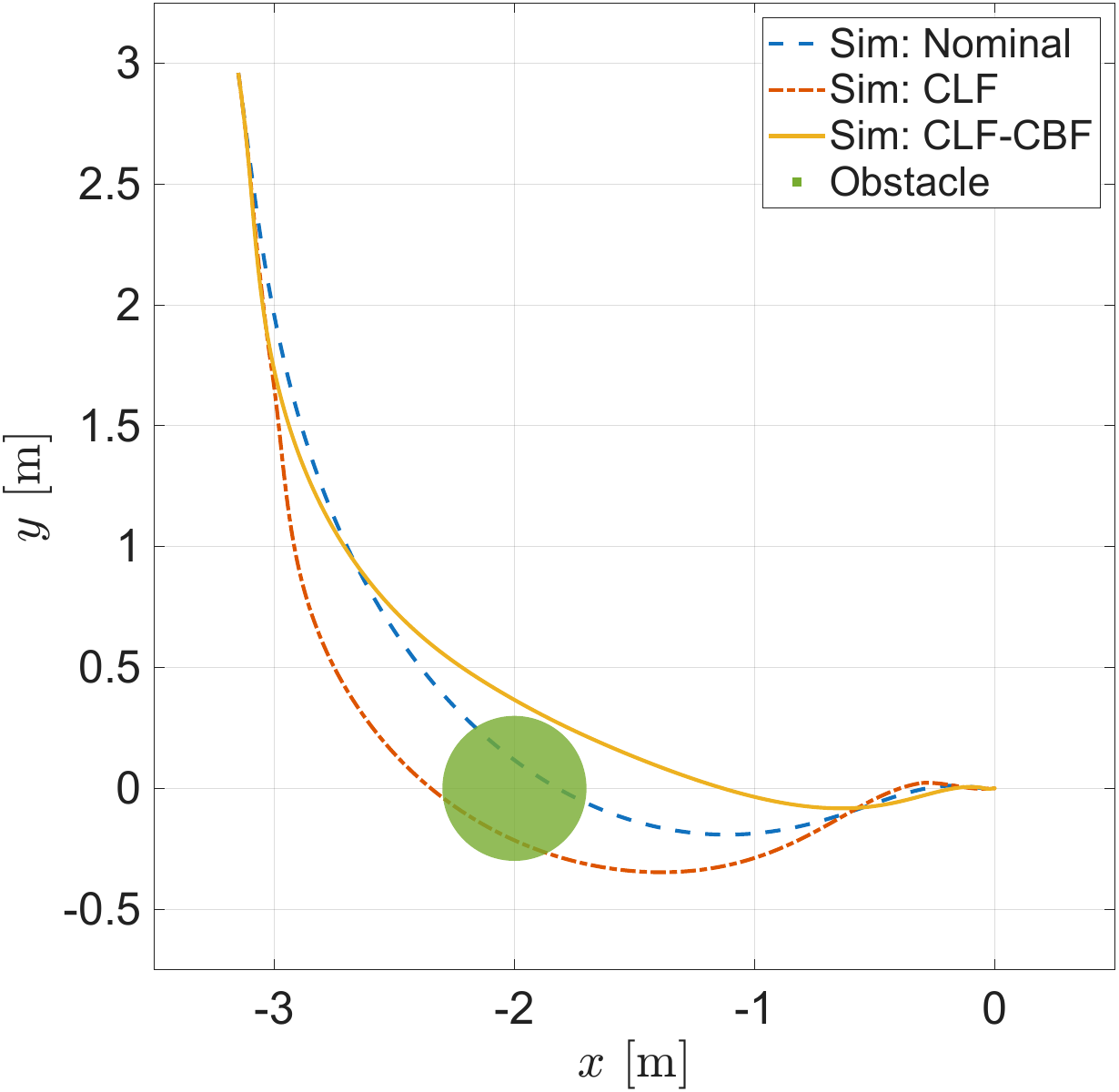}
    \caption{Simulation paths of the robot in the XY plane.}
    \label{fig:sim_path}
\end{figure}

\begin{figure}[t]
    \centering
    \includegraphics[width=0.8\linewidth]{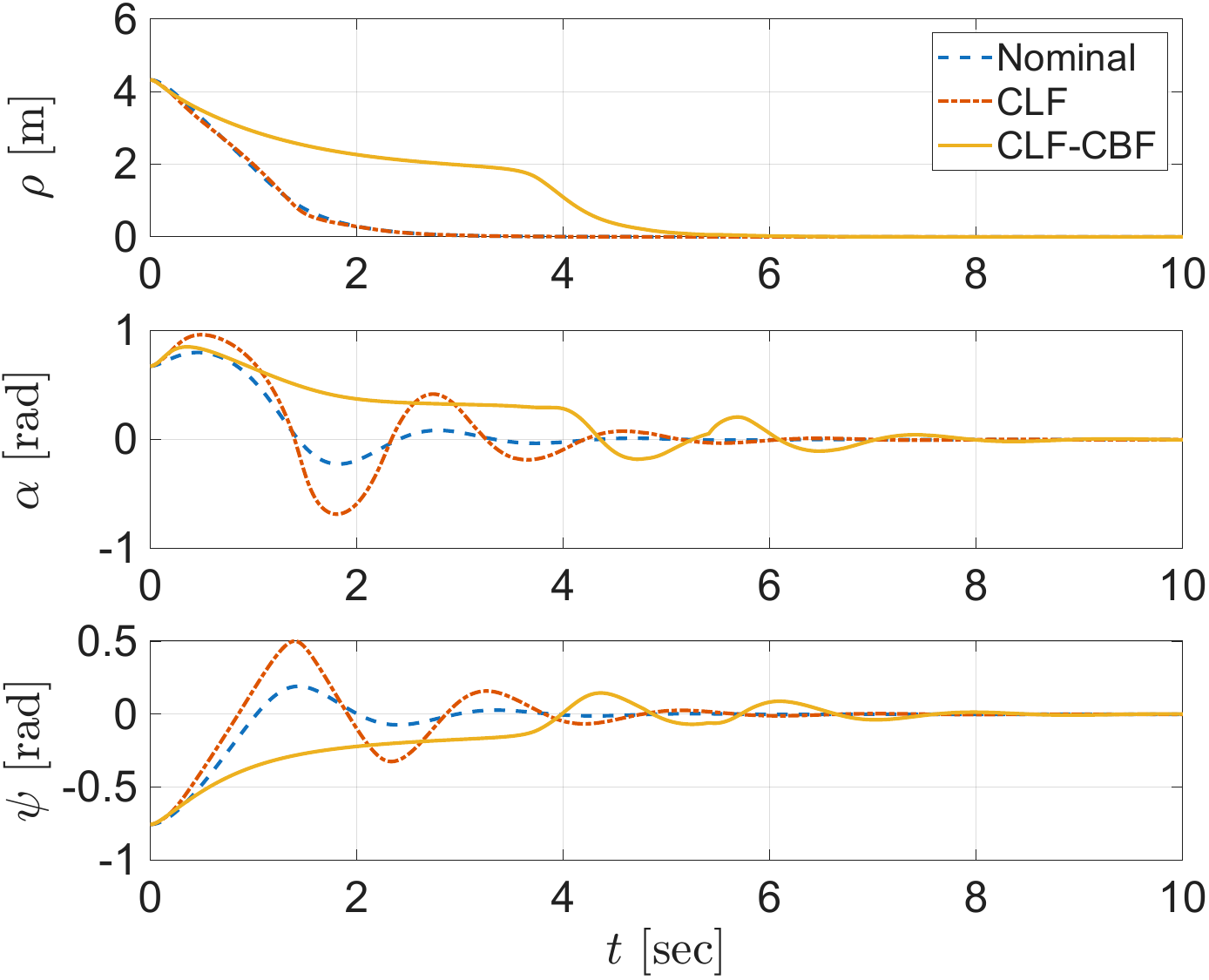}
    \caption{Simulation trajectories of the robot in polar coordinates.}
    \label{fig:sim_traj}
\end{figure}

\subsection{Experimental Results}

The experiments were conducted in the Autonomous Systems and Control Laboratory (ASCL) at the City College of New York. The experiment setup is shown in Fig. \ref{fig:exp_frame}. The experimental setup comprises a differential-drive nonholonomic mobile robot operating within a 6 m $\times$ 6 m workspace. High-precision global localization is achieved using a VICON motion capture system equipped with eight Vero 2.2 cameras, operating at 330 Hz with an accuracy of 1 mm. The computational architecture consists of a host PC for data processing and data streaming, and a laptop dedicated to executing the proposed control algorithm. The proposed safety-critical stabilization algorithm is implemented on the laptop using MATLAB/Simulink R2025a. An unpowered robot was strategically placed at $(-0.6,0.4)$ to serve as a static obstacle to evaluate avoidance capabilities. The admissible set is given by $C_0:=\{(x,y)\in\mathbb{R}^2:h_0(x,y)=40((x + 0.6)^2 + (y-0.4)^2 - 0.2^2)\}$. The robot was initially at $(x_0,y_0,\theta_0)=(-1.08,1.37,0.78)$, with both linear and angular velocities set to zero. The target position was set at the origin. The same control parameters as in the simulations were used. The experimental results are shown in Figs. \ref{fig:exp_path}-\ref{fig:exp_traj}, which illustrate that the proposed CLF-CBF $\gamma m$-QP controller successfully performs parking while avoiding obstacles. It should be noted that the experimental trajectory differs from the simulation trajectory in Figs. \ref{fig:exp_path}-\ref{fig:exp_traj}, and the angular error does not converge exactly to zero. This discrepancy is primarily due to actuator saturation and the dead-zone effect when the control input is small.

\begin{remark}\rm
    In the experiments, a safety margin was deliberately introduced to account for the physical size of the mobile robot. Specifically, the effective obstacle boundaries used in the CBF ($r=0.2$) were defined to be slightly larger than the actual physical obstacles. This ensures that even when the robot’s center of mass reaches the boundary of the defined obstacle, the robot’s physical body still remains clear of contact with the real obstacle.
\end{remark}

\begin{figure}[t]
    \centering
    \includegraphics[width=0.95\linewidth]{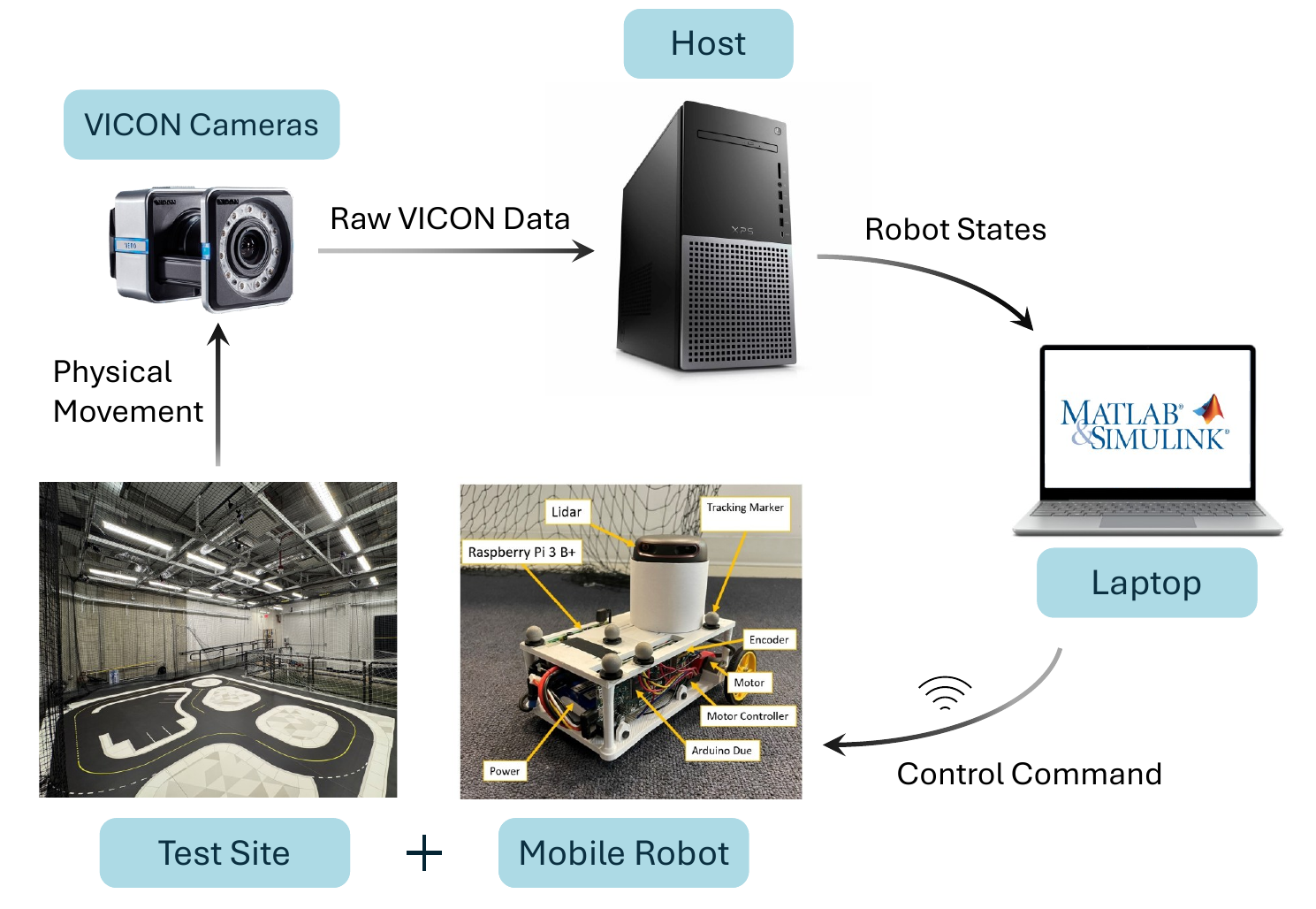}
    \caption{Experimental system framework.}
    \label{fig:exp_frame}
\end{figure}

\begin{figure}[t]
    \centering
    \includegraphics[width=0.7\linewidth]{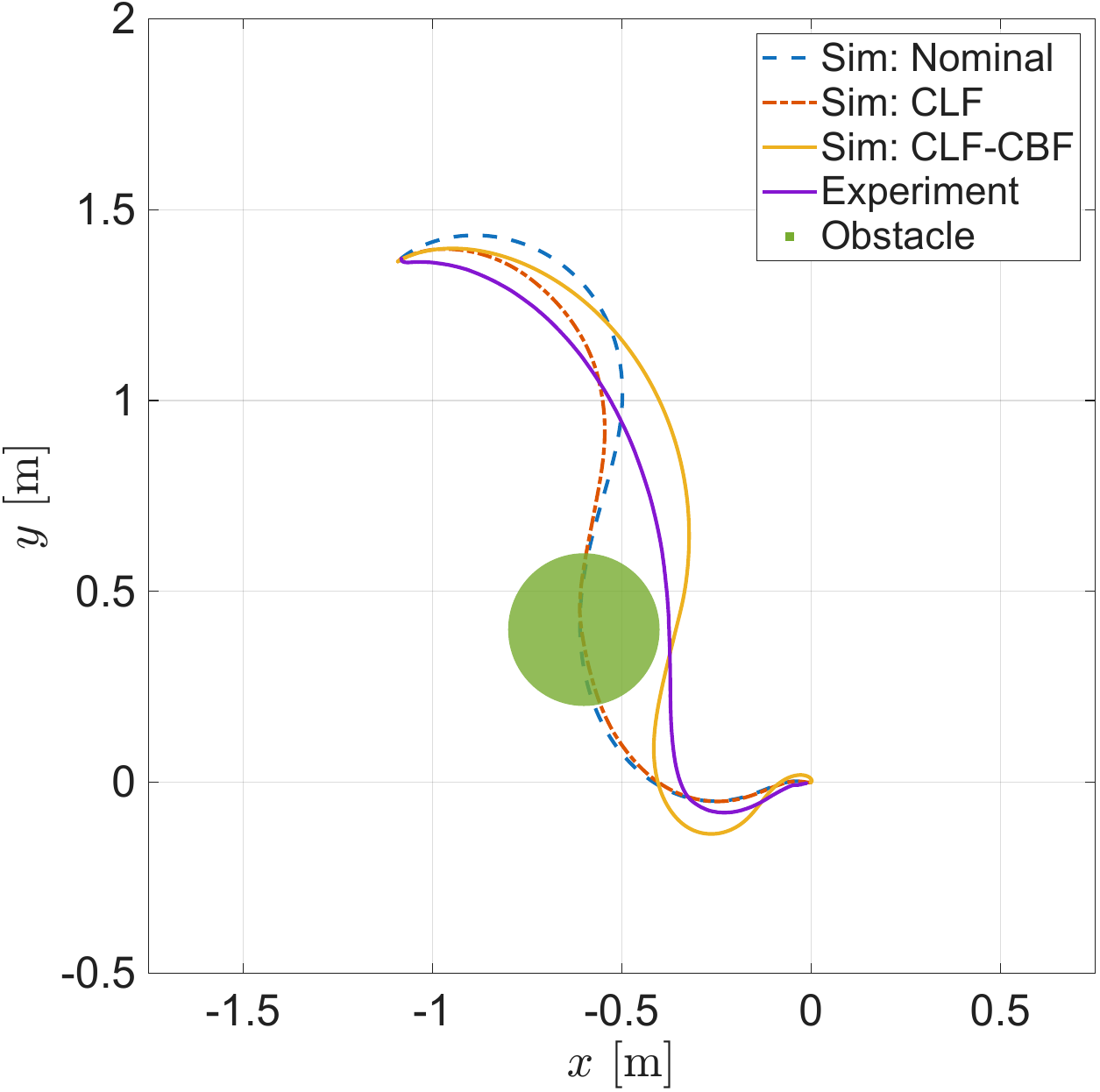}
    \caption{Paths of the robot in the XY plane.}
    \label{fig:exp_path}
\end{figure}

\begin{figure}[t]
    \centering
    \includegraphics[width=0.8\linewidth]{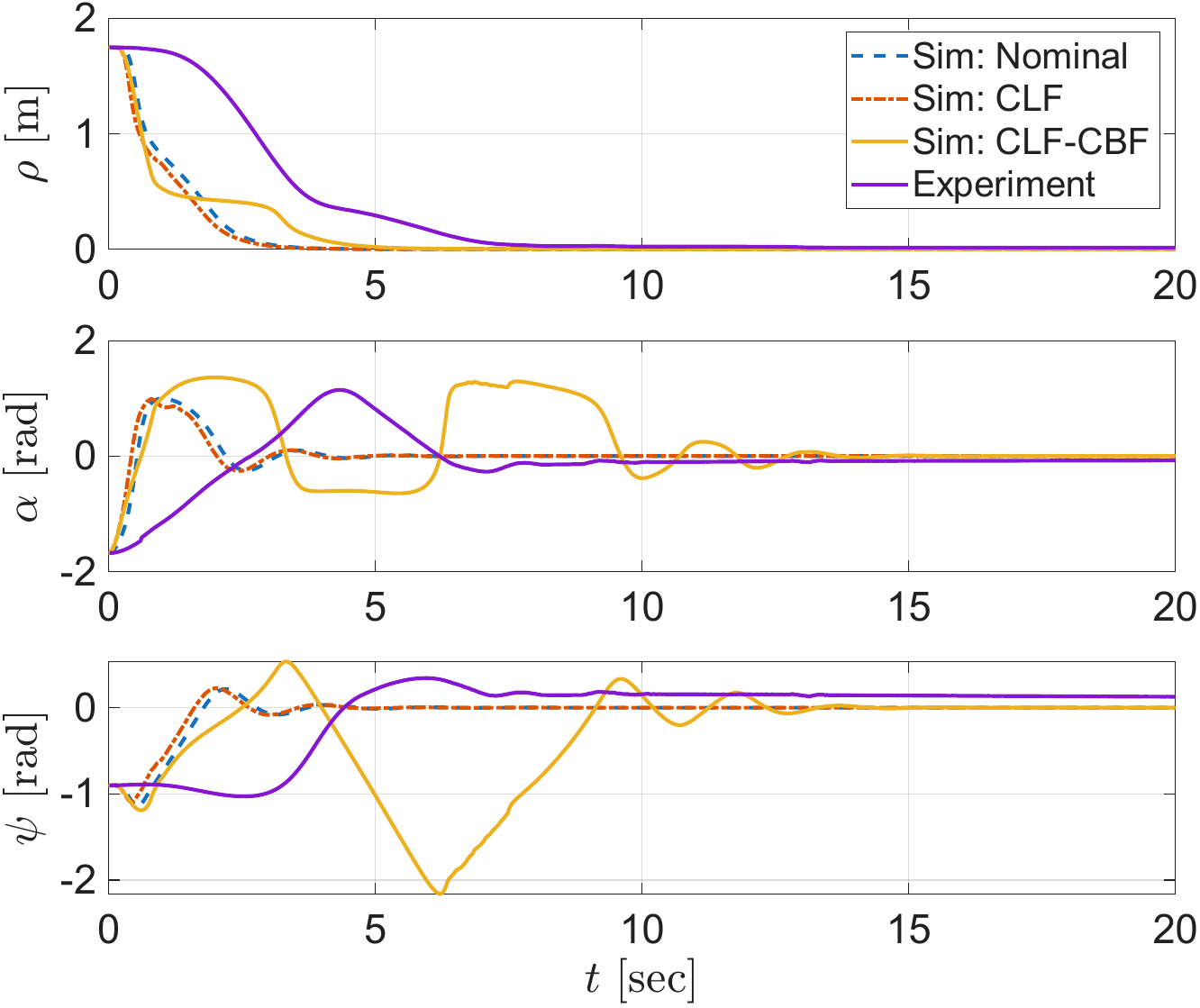}
    \caption{Trajectories of the robot in polar coordinates.}
    \label{fig:exp_traj}
\end{figure}

\section{Conclusions}\label{sec:conclusions}

This work presents a continuous, time-invariant control strategy grounded in the $\gamma m$-QP framework, which integrates CLFs and CBFs to ensure both stability and safety for the closed-loop system. Notably, we develop a global, time-invariant, strict Lyapunov function for a nonholonomic mobile robot system, utilizing a nominal stabilization controller in polar coordinates. This strict Lyapunov function is subsequently employed as the global CLF in the QP formulation. Furthermore, by leveraging the inherent cascaded structure of the vehicle's dynamics, we construct a CBF for the mobile robot through an integrator backstepping approach. The main results guarantee that the closed-loop system achieves both asymptotic stability and safety. Experimental validations are provided to demonstrate the efficacy and performance of the proposed method. Future research will focus on extending this framework to address safety formation control in multi-agent systems, incorporating robustness analysis and explicitly accounting for input saturation.

\section*{Acknowledgment} 
{The first author would like to express his gratitude to Prof.\ Miroslav Krsti\'c of UCSD for the fruitful discussions regarding the construction of the global strict Lyapunov function and Remarks~\ref{rmk:2} and~\ref{rmk:3}. The authors would also like to thank the Associate Editor and the anonymous reviewers for their useful suggestions.}

The work of Bo Wang was supported in part by GSoE at CCNY, and in part by the PSC-CUNY Award, jointly funded by The Professional Staff Congress and The City University of New York. 

\section*{Conflict of Interest}
There are no conflicts of interest.

\section*{Data Availability Statement}
The datasets generated and supporting the findings of this article are obtainable from the corresponding author upon reasonable request.

\appendix
\section{Proof of Proposition \ref{prop:1}}\label{app:A}

First, we show that the function \( V \) is positive definite and proper. Since \( U \) is a positive definite quadratic form, it suffices to show that the function \( W^\sharp \) is positive definite and proper in its arguments. Direct calculation yields
\[
\operatorname{det}(P) = \frac{k_\alpha^2 + k_\rho^2 \lambda^2 + 2k_\rho^2 \lambda + k_\rho^2}{4k_\alpha^2 k_\rho^2 \lambda} > 0.
\]
Hence, the matrix \( P = P^\top>0 \), implying that \( W_2 \) is a positive definite quadratic form. Since \( W_1 \) is also a positive definite quadratic form, it follows that \( W \) is positive definite and proper. Consequently, \( W^\sharp \) is also positive definite and proper in its arguments.

Next, according to the definition of a CLF, to establish that \( V \) is a CLF, we need to show that for all \( (\rho, \alpha, \psi, z, \tilde{\omega}) \neq 0 \), there exists a control input \( (u_v, u_\omega) \) such that 
$\dot{V}|_{\eqref{eq:kinematics}, \eqref{eq:dynamics}} < 0$.
We demonstrate this by explicitly constructing a nominal control law \( (u_v, u_\omega) \).

In the new velocity coordinates $\tilde{\omega}:=\omega-\omega^*$, $z:=(v-v^*)/{\rho}$, the kinematics \eqref{eq:kinematics} become
\begin{equation}\label{eq:14}
    \begin{bmatrix}
        \dot{\rho}\\ \dot{\alpha} \\ \dot{\psi}
    \end{bmatrix}=
    \underbrace{\begin{bmatrix}
        -k_\rho\cos(\alpha)^2 \rho \\ -k_\alpha \alpha - k_\rho\operatorname{sinc}(2\alpha)\lambda\psi\\ k_\rho\operatorname{sinc}(2\alpha)\alpha
    \end{bmatrix}}_{f_{\text{nom}}(\rho,\alpha,\psi)}+
    \underbrace{\begin{bmatrix}
        -\rho\cos(\alpha) & 0 \\
        \sin(\alpha) & 1 \\
        \sin(\alpha) & 0
    \end{bmatrix}}_{g_{\text{nom}}(\rho,\alpha)}
    \begin{bmatrix}
        z\\ \tilde{\omega}
    \end{bmatrix}.
\end{equation}

Also, the velocity dynamics in the new coordinates are given by
\begin{equation}\label{eq:115}
    \begin{cases}
        \dot{z}=\dfrac{1}{\rho}(u_v-\dot{v}^*)+k_\rho\cos(\alpha)^2 z + \cos(\alpha)z^2,\\
        \dot{\tilde{\omega}}=u_\omega-\dot{\omega}^*.
    \end{cases}
\end{equation}
The nominal control law \( (u_v, u_\omega) \) can be selected as the feedback linearization control law
\begin{equation}\label{eq:16}
    \begin{cases}
        u_v=\dot{v}^*-\rho\left[k_\rho\cos(\alpha)^2 z + \cos(\alpha)z^2+k_z z \right],\\
        u_\omega=\dot{\omega}^*-k_\omega \tilde{\omega},
    \end{cases}
\end{equation}
which yields the linear closed-loop velocity dynamics
\begin{equation}\label{eq:17}
    \dot{z}=-k_z z,\quad \dot{\tilde{\omega}}=-k_\omega \tilde{\omega}.
\end{equation}
Next, we show that $\dot{V}|_{\eqref{eq:14},\eqref{eq:17}}<0$ for all \( (\rho, \alpha, \psi, z, \tilde{\omega}) \neq 0 \). 

Noting that the nominal closed-loop system \eqref{eq:14}, \eqref{eq:17} exhibits a cascaded structure, we first consider the subsystem \eqref{eq:14} restricted to the manifold \( \{ z = \tilde{\omega} = 0 \} \). Evaluating the total derivative of \( W_1 \) along the vector field \( f_{\text{nom}} \) in \eqref{eq:14} yields
\begin{equation}
    \dot{W}_1|_{f_{\text{nom}}}:=\langle \nabla W_1, f_{\text{nom}} \rangle = -k_\rho \cos^2(\alpha) \rho^2 - k_\alpha \alpha^2 \le 0,
\end{equation}
where $\nabla$ represents the gradient and $\langle\cdot,\cdot\rangle$ represents the inner product. Then, by adding and subtracting the terms $-k_\rho\lambda\psi$ and $k_\rho\alpha$ in second and third rows in $f_{\text{nom}}$, respectively, the $\alpha$- and $\psi$-dynamics restricted to the manifold \( \{ z = \tilde{\omega} = 0 \} \) are given by
\begin{equation}\label{eq:19}
    \begin{bmatrix}
        \dot{\alpha} \\ \dot{\psi}
    \end{bmatrix}=
    \underbrace{\begin{bmatrix}
        -k_\alpha & -k_\rho\lambda \\ k_\rho & 0
    \end{bmatrix}}_{A}
    \underbrace{\begin{bmatrix}
        \alpha \\ \psi
    \end{bmatrix}}_{\xi}
    +
    \underbrace{\begin{bmatrix}
        -\lambda k_\rho (\operatorname{sinc}(2\alpha)-1) \psi \\
        k_\rho (\operatorname{sinc}(2\alpha)-1) \alpha
    \end{bmatrix}}_{K(\alpha,\psi)}.
\end{equation}
Since the matrix $A$ in \eqref{eq:19} is Hurwitz, the Lyapunov equation $A^\top P + PA = -I$ has a unique, positive definite solution $P$, which is given in Proposition \ref{prop:1}. In other words, denoting $\xi:=[\alpha~\psi]^\top$, $W_2$ is a strict Lyapunov function for the linear system $\dot{\xi}=A\xi$, i.e., $\langle W_2,A\xi \rangle=-|\xi|^2$. It is easy to show that $|\operatorname{sinc}(2s)-1|\le \frac{2}{\pi}|s|$, and thus, we have $|K(\alpha,\psi)|\le \frac{2}{\pi}k_\rho \lambda |\alpha||\xi|$. 
\footnote{Note that $\displaystyle{\sup_{s\in\mathbb{R}}}\dfrac{\sin(s)-s}{s^2}=\dfrac{1}{\pi}$.}
The total derivative of $W_2$ along trajectories of \eqref{eq:19} is then given by
\begin{eqnarray}
    \dot{W}_2|_{\eqref{eq:19}}&=&-|\xi|^2+2\xi^\top PK(\alpha,\psi) \label{eq:20}\\
    &\le & -|\xi|^2+\frac{2}{\pi}k_\rho \lambda \lambda_M(P)|\xi| (2|\alpha|\cdot|\xi|) \notag \\
    &\le & -|\xi|^2+\frac{2}{\pi}k_\rho \lambda \lambda_M(P)|\xi| \left(\varepsilon|\xi|^2+\frac{\alpha^2}{\varepsilon}\right),
\end{eqnarray}
where the last inequality is due to Young's inequality, and $\varepsilon>0$ can be chosen as an arbitrary positive number. Hence, for $|\xi|\ne 0$, letting $\varepsilon:=\frac{\pi}{4k_\rho\lambda \lambda_M(P)|\xi|}>0$, it follows that 
\begin{eqnarray}
    \dot{W}_2|_{\eqref{eq:19}} &\le & -\frac{1}{2}|\xi|^2 +\frac{8}{\pi^2}k_\rho^2 \lambda^2 \lambda_M^2(P)|\xi|^2\alpha^2 \notag \\
    &\le & -\frac{1}{2}|\xi|^2 +k_\alpha Q(W_1(\rho,\alpha,\psi))\alpha^2,\label{eq:21}
\end{eqnarray}
where in the last inequality we use $|\xi|^2\le 2W_1(\rho,\alpha,\psi)$. For $|\xi|=0$, it follows from \eqref{eq:20} that \eqref{eq:21} is also true. Consequently, we have 
\begin{equation}
    \dot{W}|_{f_{\text{nom}}}:=\langle \nabla W,f_{\text{nom}} \rangle \le -\frac{1}{2}|\xi|^2 - k_\rho\cos^2(\alpha)\rho^2<0.
\end{equation}
That is, $W$ is a global, strict Lyapunov function for the subsystem \eqref{eq:14} restricted to the manifold \( \{ z = \tilde{\omega} = 0 \} \). One can easily prove that $W^\sharp$ is also a global, strict Lyapunov function for the subsystem \eqref{eq:14} restricted to the manifold \( \{ z = \tilde{\omega} = 0 \} \). 

Denoting $\zeta:=[z~\tilde{\omega}]^\top$, direct calculation shows that
\begin{eqnarray}\label{eq:A11}
    \dot{V}|_{\eqref{eq:14},\eqref{eq:17}}&=&\frac{\mu W\dot{W}|_{f_{\text{nom}}}}{(W+1)^2}+\frac{\mu W}{W+1}\frac{ L_{g_{\text{nom}}}W}{{W+1}}\cdot\zeta-|\zeta|^2.
\end{eqnarray}
Note that the first and third terms on the right-hand side of \eqref{eq:A11} are negative definite terms, while the second term is indefinite. In the second term, ${L_{g_{\text{nom}}}W}/{({W+1})}$ is globally bounded, i.e., $\exists c>0$ such that $|{ L_{g_{\text{nom}}}W}/{({W+1})}|\le c$. Hence, together with Young's inequality we have that
\begin{eqnarray}
    \frac{\mu W}{W+1}\frac{ L_{g_{\text{nom}}}W}{{W+1}}\cdot\zeta&\le &  \frac{c\mu W}{W+1}\cdot|\zeta|\\
    &\le & \frac{c^2\mu^2}{2}\left(\frac{ W}{W+1}\right)^2+\frac{1}{2}|\zeta|^2. \label{eq:A13}
\end{eqnarray}
The term $-|\zeta|^2$ in \eqref{eq:A11} dominates the term $\frac{1}{2}|\zeta|^2$ in \eqref{eq:A13}. Moreover, the term $\frac{\mu W \dot{W}|_{f_{\text{nom}}}}{(W+1)^2}$ in \eqref{eq:A11} dominates the term $\frac{c^2\mu^2}{2} \left( \frac{W}{W+1} \right)^2$ in \eqref{eq:A13} near the origin, since the latter has a higher degree. Away from the origin, there exists a sufficiently small $\mu > 0$ such that $\frac{\mu W \dot{W}|_{f_{\text{nom}}}}{(W+1)^2}$ continues to dominate $\frac{c^2\mu^2}{2} \left( \frac{W}{W+1} \right)^2$, due to the fact that $W/(W+1)$ remains globally bounded. Therefore, we conclude that $\dot{V}|_{\eqref{eq:14},\eqref{eq:17}} < 0$ is negative definite.

Finally, we conclude the proof by noting that \eqref{eq:16} is continuous, and $|(u_v, u_\omega)|\to 0$ as $|(\rho, \alpha, \psi, z, \tilde{\omega})|\to 0$, which establishes the small control
property.\qed

\bibliographystyle{asmejour} 

\bibliography{asmejour} 

\begin{thebibliography}{10}
\newcommand{\enquote}[1]{``#1''}
\providecommand{\url}[1]{\texttt{#1}}
\providecommand{\urlprefix}{}
\expandafter\ifx\csname urlstyle\endcsname\relax
  \providecommand{\doi}[1]{doi:\discretionary{}{}{}#1}\else
  \providecommand{\doi}{doi:\discretionary{}{}{}\begingroup \urlstyle{rm}\Url}\fi
\providecommand{\eprint}[2][]{\urlprefix\url{#1#2}}
\providecommand{\hrefurl}[2][]{\href{#1}{#2}}

\bibitem{kolmanovsky1995developments}
Kolmanovsky, I. and McClamroch, N.~H., 1995, \enquote{Developments in nonholonomic control problems,} IEEE Control Syst., \textbf{15}(6), pp. 20--36.

\bibitem{brockett1983asymptotic}
Brockett, R., 1983, \enquote{Asymptotic stability and feedback stabilization,} \textit{Differential geometric control theory}, R.~S. Millman, R.~W. Brocket, and H.~J. Sussmann, eds., Birkh{\"a}user, pp. 181--191.

\bibitem{maghenem2019consensus}
Maghenem, M., Bautista, A., Nu{\~n}o, E., Lor{\'\i}a, A., and Panteley, E., 2019, \enquote{Consensus of multi-agent systems with nonholonomic restrictions via {L}yapunov’s direct method,} IEEE Contr. Syst. Lett., \textbf{3}(2), pp. 344--349.

\bibitem{han2024safety}
Han, T. and Wang, B., 2024, \enquote{Safety-Critical Stabilization of Force-Controlled Nonholonomic Mobile Robots,} IEEE Contr. Syst. Lett., \textbf{8}, pp. 2469--2474.

\bibitem{todorovski2025modular}
Todorovski, V., Kim, K.~H., and Krsti\'c, M., 2025, \enquote{Modular design of strict control {L}yapunov functions for global stabilization of the unicycle in polar coordinates,} arXiv preprint arXiv:2509.25575.

\bibitem{krstic2025integrator}
Krsti\'c, M., Todorovski, V., Kim, K.~H., and Astolfi, A., 2025, \enquote{Integrator Forwading Design for Unicycles with Constant and Actuated Velocity in Polar Coordinates,} arXiv preprint arXiv:2509.25579.

\bibitem{wang2022robust}
Wang, B., Nersesov, S.~G., and Ashrafiuon, H., 2022, \enquote{Robust formation control and obstacle avoidance for heterogeneous underactuated surface vessel networks,} IEEE Trans. Control Netw. Syst., \textbf{9}(1), pp. 125--137.

\bibitem{jankovic2024multiagent}
Jankovic, M., Santillo, M., and Wang, Y., 2024, \enquote{Multiagent Systems With {CBF}-Based Controllers: Collision Avoidance and Liveness From Instability,} IEEE Trans. Control Syst. Technol., \textbf{32}(2), pp. 705--712.

\bibitem{ames2014control}
Ames, A.~D., Grizzle, J.~W., and Tabuada, P., 2014, \enquote{Control barrier function based quadratic programs with application to adaptive cruise control,} \textit{Proc. IEEE Conf. Decis. Control}, IEEE, pp. 6271--6278.

\bibitem{ames2017control}
Ames, A.~D., Xu, X., Grizzle, J.~W., and Tabuada, P., 2017, \enquote{Control barrier function based quadratic programs for safety critical systems,} IEEE Trans. Autom. Contr., \textbf{62}(8), pp. 3861--3876.

\bibitem{ames2019control}
Ames, A.~D., Coogan, S., Egerstedt, M., Notomista, G., Sreenath, K., and Tabuada, P., 2019, \enquote{Control barrier functions: Theory and applications,} \textit{Proc. Euro. Control Conf.}, IEEE, pp. 3420--3431.

\bibitem{xu2018correctness}
Xu, X., Grizzle, J.~W., Tabuada, P., and Ames, A.~D., 2018, \enquote{Correctness guarantees for the composition of lane keeping and adaptive cruise control,} IEEE Trans. Autom. Sci. Eng., \textbf{15}(3), pp. 1216--1229.

\bibitem{clark2021control}
Clark, A., 2021, \enquote{Control barrier functions for stochastic systems,} Automatica, \textbf{130}, p. 109688.

\bibitem{xu2015robustness}
Xu, X., Tabuada, P., Grizzle, J.~W., and Ames, A.~D., 2015, \enquote{Robustness of control barrier functions for safety critical control,} IFAC-PapersOnLine, \textbf{48}(27), pp. 54--61.

\bibitem{gurriet2018towards}
Gurriet, T., Singletary, A., Reher, J., Ciarletta, L., Feron, E., and Ames, A., 2018, \enquote{Towards a framework for realizable safety critical control through active set invariance,} \textit{Proceedings of the 9th ACM/IEEE International Conference on Cyber-Physical Systems}, IEEE, pp. 98--106.

\bibitem{singletary2021safety}
Singletary, A., Kolathaya, S., and Ames, A.~D., 2021, \enquote{Safety-critical kinematic control of robotic systems,} IEEE Contr. Syst. Lett., \textbf{6}, pp. 139--144.

\bibitem{jankovic2018robust}
Jankovic, M., 2018, \enquote{Robust control barrier functions for constrained stabilization of nonlinear systems,} Automatica, \textbf{96}, pp. 359--367.

\bibitem{glotfelter2019hybrid}
Glotfelter, P., Buckley, I., and Egerstedt, M., 2019, \enquote{Hybrid nonsmooth barrier functions with applications to provably safe and composable collision avoidance for robotic systems,} IEEE Robot. Autom. Lett., \textbf{4}(2), pp. 1303--1310.

\bibitem{xiao2021high}
Xiao, W. and Belta, C., 2021, \enquote{High-order control barrier functions,} IEEE Trans. Autom. Contr., \textbf{67}(7), pp. 3655--3662.

\bibitem{alan2023control}
Alan, A., Taylor, A.~J., He, C.~R., Ames, A.~D., and Orosz, G., 2023, \enquote{Control barrier functions and input-to-state safety with application to automated vehicles,} IEEE Trans. Control Syst. Technol., \textbf{31}(6), pp. 2744--2759.

\bibitem{taylor2022safe}
Taylor, A.~J., Ong, P., Molnar, T.~G., and Ames, A.~D., 2022, \enquote{Safe backstepping with control barrier functions,} \textit{Proc. IEEE Conf. Decis. Control}, IEEE, pp. 5775--5782.

\bibitem{haraldsen2024safety}
Haraldsen, A., Wiig, M.~S., Ames, A.~D., and Pattersen, K.~Y., 2024, \enquote{Safety-critical control of nonholonomic vehicles in dynamic environments using velocity obstacles,} \textit{Proc. Amer. Contr. Conf.}, IEEE, pp. 3152--3159.

\bibitem{maghenem2017cascades}
Maghenem, M., Lor{\'\i}a, A., and Panteley, E., 2017, \enquote{A cascades approach to formation-tracking stabilization of force-controlled autonomous vehicles,} IEEE Trans. Autom. Contr., \textbf{63}(8), pp. 2662--2669.

\bibitem{wang2022time}
Wang, B., Nersesov, S.~G., and Ashrafiuon, H., 2022, \enquote{Time-Varying Formation Control for Heterogeneous Planar Underactuated Multivehicle Systems,} ASME J. Dyn. Syst. Meas. Contr., \textbf{144}(4), p. 041006.

\bibitem{reis2020control}
Reis, M.~F., Aguiar, A.~P., and Tabuada, P., 2020, \enquote{Control barrier function-based quadratic programs introduce undesirable asymptotically stable equilibria,} IEEE Contr. Syst. Lett., \textbf{5}(2), pp. 731--736.

\bibitem{sontag1998mathematical}
Sontag, E.~D., 1998, \textit{Mathematical Control Theory}, 2nd ed., Springer-Verlag, New York, NY.

\bibitem{wang2021formation}
Wang, B., Nersesov, S., and Ashrafiuon, H., 2021, \enquote{Formation Regulation and Tracking Control for Nonholonomic Mobile Robot Networks Using Polar Coordinates,} IEEE Contr. Syst. Lett., \textbf{6}, pp. 1909--1914.

\bibitem{cohen2024safety}
Cohen, M.~H., Molnar, T.~G., and Ames, A.~D., 2024, \enquote{Safety-critical control for autonomous systems: Control barrier functions via reduced-order models,} Annu. Rev. Control, \textbf{57}, p. 100947.

\bibitem{freeman1996robust}
Freeman, R.~A. and Kokotovi\'c, P.~V., 1996, \textit{Robust Nonlinear Control Design}, Birkhäuser, Boston.

\end{thebibliography}

\end{document}